\documentclass{llncs}

\usepackage[T1]{fontenc}
\usepackage[latin2]{inputenc}
\usepackage[english]{babel}

\usepackage{dsfont, complexity} 
\usepackage{amsmath}
\usepackage{amssymb}
\usepackage{lmodern}
\usepackage{float}
\usepackage{graphicx}
\usepackage{tikz}
\usepackage{xcolor}
\usepackage{paralist}
\usetikzlibrary{arrows,decorations.markings,patterns,calc, positioning}

\usepackage{algpseudocode}
\usepackage[ruled]{algorithm} 

\newcommand{\mycaption}[2]
 {\begin{center} \parbox{4in}{\caption{\small #2 \label{#1}}} \end{center}}

\begin{document}

\title{Improved Algorithmic Results for Unsplittable Stable Allocation 
Problems}
\author{\'{A}gnes Cseh\inst{1} \and Brian C. Dean\inst{2}}

\institute{Institute for Mathematics, TU Berlin
\and School of Computing, Clemson University}

\maketitle

\begin{abstract}
The stable allocation problem is a many-to-many generalization of the
well-known stable marriage problem, where we seek a bipartite
assignment between, say, jobs (of varying sizes) and machines (of
varying capacities) that is ``stable'' based on a set of underlying
preference lists submitted by the jobs and machines.  Building on the
initial work of \cite{dean_unsplit}, we study a natural
``unsplittable'' variant of this problem, where each assigned job must
be fully assigned to a single machine.  Such unsplittable bipartite
assignment problems generally tend to be NP-hard, including
previously-proposed variants of the unsplittable stable allocation
problem \cite{DBLP:journals/jco/McDermidM10}.  Our main result is to
show that under an alternative model of stability, the unsplittable
stable allocation problem becomes solvable in polynomial time;
although this model is less likely to admit feasible solutions than
the model proposed in \cite{DBLP:journals/jco/McDermidM10}, we show
that in the event there is no feasible solution, our approach computes
a solution of minimal total congestion (overfilling of all machines
collectively beyond their capacities).  We also describe a technique
for rounding the solution of a stable allocation problem to produce
``relaxed'' unsplit solutions that are only mildly infeasible, where
each machine is overcongested by at most a single job.  
\end{abstract}




\section{Introduction}
Consider a bipartite assignment problem over a graph $G = (V = J \cup
M, E)$ involving the assignment of a set of jobs $J$ to a set of
machines $M$. Each job $j \in J$ has a processing time $q(j)$, 
each machine $m \in M$ has a capacity $q(m)$, and there is a capacity
$c(jm)$ for each edge $jm \in E$ governing the maximum amount of job
$j$ that can be assigned to machine $m$.  A feasible assignment of
jobs to machines is described by a function $x : E \rightarrow
\mathbb{R}_{\geq 0}$ such that 
\begin{enumerate}
\item $0 \leq x(jm) \leq c(jm)$ for all edges $jm \in E$,
\item $x(j) := \sum_{m \in M} {x(jm)} \leq q(j)$ for all jobs $j \in J$, and
\item $x(m) := \sum_{j \in J} {x(jm)} \leq q(m)$ for all machines $m \in M$.
\end{enumerate}
If $x(jm) \in \{0, q(j)\}$ for all $jm \in E$, we say the assignment
is \emph{unsplit}, since each assigned job is assigned in its entirety
to a single machine. We often forgo the use of edge capacities $c(jm)$
when discussing unsplit assignments, since an edge $jm$ can simply be
deleted if $c(jm) < q(j)$.

Problems of the form above have been extensively studied in the
algorithmic literature, where typical objectives are to find a
feasible assignment or one of maximum weight (maximizing a linear
objective function $\sum_{jm \in E} w(jm) x(jm)$, where $w(jm)$ is the
weight of edge $jm \in E$).  While the fractional (splittable)
variants of these problems are easy to solve in polynomial time via
network flow techniques, the unsplittable variants are somewhat more
interesting.  In linear time, one can greedily assign jobs arbitrarily
to machines until no further assignments are possible, thereby
producing a {\em maximal} assignment.  However, if we care about
finding an unsplit assignment of measurably good quality, the problem
becomes substantially harder.  It is NP-hard to find an unsplit
assignment of either maximum total size $|x| = \sum_{jm \in E} x(jm)$
or of maximum weight; the former is a variant of the multiple subset
sum problem \cite{Caprara2000111}, and the latter is known as the
multiple knapsack problem \cite{DBLP:journals/siamcomp/ChekuriK05}.

In contrast to problems with explicit edge costs, the \emph{stable
  allocation problem} is an ``ordinal'' problem variant where the
quality of an assignment is expressed in a more game theoretic setting
via ranked preference lists submitted by the jobs and machines, with
respect to which we seek an assignment that is \emph{stable} (defined
shortly).  In this paper, we study the stable allocation problem in
the unsplittable setting, which was shown to be NP-hard in
\cite{DBLP:journals/jco/McDermidM10} using one natural definition for
stability.  We show here that by contrast, a different and more strict
notion of stability, proposed initially in \cite{dean_unsplit}, leads
to an $O(|E|)$ algorithm for the unsplit problem.  The tradeoff is
that under this different notion of stability, it is unlikely that
feasible solutions will exist.  However, we show that by relaxing the
problem to allow mildly infeasible solutions, our algorithm 
computes a ``relaxed'' unsplit stable solution (in~which each machine
is filled beyond its capacity by at most the allocation of a single
job) in which the total amount of overcongestion across all machines,
$\sum_{m \in M} { \max\left(0, x(m)-q(m)\right)}$, is minimized (so
in particular, if there is a feasible solution with no congestion, we
will find it).

Through the work of several former authors
\cite{DinitzGG99,Skutella00,ShmoysT93}, the ``relaxed'' model has
become relatively popular in the context of unsplittable bipartite
assignment and unsplittable flow problems.  The standard approximation
algorithm framework (finding an approximately-optimal, feasible
solution) typically does not fit these problems, since finding any
feasible solution is typically NP-hard.  Instead, authors tend to
focus on pseudo-approximation results with minimal congestion per
machine or per edge.  Analogous results were previously developed for
unsplit stable allocation problems in \cite{dean_unsplit}, where an
unsplit stable allocation can be found in linear time in which each
machine is overcongested by at most a single job.  The model of
stability proposed in \cite{dean_unsplit} is the one we further
develop in this paper, and among all of these prior approaches
(including those for standard unsplittable bipartite assignment and
flows), it seems to be the only unsplit model studied to date in which
minimization of {\em total} congestion is possible in polynomial time.
Hence, there is a substantial algorithmic incentive to consider this
model, even though its notion of stability is less natural than in
\cite{DBLP:journals/jco/McDermidM10}.

In our ``relaxed'' unsplit model, we develop new structural and
algorithmic results by showing how to compute in $O(|E|)$ time a
``job-optimal'' assignment that maximizes the total size $|x|$ of all
assigned jobs, and a ``machine-optimal'' assignment that minimizes
$|x|$.  It is this machine-optimal solution that we show also
minimizes total congestion.  In order to produce potentially other
solutions (e.g., that might be more fair to both sides), we show also
a technique for ``rounding'' a solution of the fractional stable
allocation problem to obtain a relaxed unsplit solution.  Finally, we
comment on several mathematical properties of the set of all relaxed
unsplit solutions, showing that while they unfortunately seem to lack
the nice distributed lattice structure satisfied by solutions of the
stable matching and allocation problems, they do at least adhere to a
weakened form of the so-called ``rural hospital'' theorem, defined
shortly.

\section{Background and Preliminaries}

\subsection{Stable Matching and Allocation Problems}

{\bf Stable Marriage.} The stable marriage (or stable matching)
problem takes place on a bipartite graph with men on one side and
women on the other, where each individual submits a strictly-ordered,
but possibly incomplete preference list of the members of the opposite
sex.  The goal is to find a matching that is \emph{stable}, containing
no \emph{blocking pair} -- an unmatched (man, woman) pair $(m,w)$
where $m$ is either unmatched or prefers $w$ to his current partner,
and likewise for $w$.  

In their seminal paper~\cite{GS:1962}, Gale and Shapley describe a
simple $O(|E|)$ algorithm to find a stable matching for any instance.
The most typical incarnation of their algorithm generates a solution
that is ``man-optimal'' and ``woman-pessimal'', where each man is
matched with the best possible partner he could receive in any stable
matching, and each woman is matched with the worst possible partner
she could receive in any stable matching.  By reversing the roles of
the men and women, the algorithm can also generate a solution that is
simultaneously woman-optimal and man-pessimal.

\noindent {\bf Stable Allocation.} The stable allocation problem was
introduced by Ba\"iou and Balinski \cite{DBLP:journals/mor/BaiouB02}
as a high-multiplicity variant of the stable matching problem, where
we match non-unit elements with non-unit elements -- here, we speak of
matching jobs of varying size with machines of varying capacity.  Just
as before, jobs and machines submit strict preferences over their
outgoing edges in the bipartite assignment graph.  If job $j \in J$
prefers machine $m_1 \in M$ to machine $m_2 \in M$, we write
$\text{rank}_j(jm_1) > \text{rank}_j(jm_2)$.  A stable allocation in
this setting is a feasible allocation (as defined in the introduction)
where for every edge $jm \in E$ with $x(jm) < c(jm)$, either $j$ is
fully assigned to machines at least as good as $m$, or $m$ is fully
assigned to jobs at least as good as~$j$.  That is, there can be no
blocking edge $jm$ where $x(jm) < c(jm)$ and both $j$ and $m$ would
prefer use more of this edge.  For sake of simplicity, we say that
edges with positive $x$ value are in~$x$. Machines with $x(m) = q(m)$
are \emph{saturated}. Later, when $x(m) > q(m)$ occurs in the relaxed
version of the problem, we talk about \emph{over-capacitated}
machines. If any job prefers machine $m$ to any of its allocated
machines, then $m$ is called \emph{popular}, otherwise $m$ is
\emph{unpopular}. Note that all popular machines must be saturated in
any stable allocation.

The stable allocation problem can be solved in $O(|E| \log |V|)$ time
\cite{DBLP:journals/algorithmica/DeanM10}.  There can be many
different solutions for the same instance, but they all have the same
total allocation $|x|$, and even stronger, the values of $x(j)$ and
$x(m)$ for each job and machine remain unchanged across all stable
allocations.  This holds for both stable marriage and stable
allocation, moreover, even for stable roommate, the non-bipartite version of the problem, and is known as the {\em rural hospital theorem}.  A
common application of stable matching in practice is the National
Resident Matching Program (NRMP), where medical school graduates in
the USA are matched with residency positions at hospitals via a
centralized stable matching procedure.  A consequence of the rural
hospital theorem is that if a less-preferred (typically rural)
hospital cannot fill its quota in some stable assignment, then there
is no stable assignment in which its quota will be filled.

Like the stable marriage problem, one can always find job-optimal,
machine-pessimal and job-pessimal, machine-optimal solutions.  To
define these notions for the stable allocation problem, Ba\"iou and
Balinski \cite{DBLP:journals/mor/BaiouB02} define an order on stable
solutions based on a {\em min-min criterion}, where a job $j$ prefers
allocation $x_1$ to allocation $x_2$ if $x_1(jm) < x_2(jm)$ implies
$x_1(jm') = 0$ for every $jm'$ worse than $jm$ for $j$.  A similar
relation can be defined for machines as well.  Stable
matchings~\cite{Knuth:1976:MSR} and stable
allocations~\cite{DBLP:journals/mor/BaiouB02} both form distributive
lattices with an ordering relation based on the min-min criterion.

\subsection{Unsplittable Stable Allocation Problems}
An unsplit allocation $x$ satisfies $x(jm) \in \{0, q(j)\}$ for all
$jm \in E$, so each job is assigned in its entirety to one machine.
For simplicity, we introduce a ``dummy'' machine $m_d$ with high
capacity, which acts as the last choice for every job.  This lets us
assume without loss of generality that an unsplittable assignment
always exists in which every job is assigned.  In this context,
we define the size $|x|$ of an assignment so that jobs assigned
to $m_d$ do not count, since they are in reality unassigned.
In addition to the application of scheduling jobs in a non-preemptive
fashion, a motivating application for the unsplittable stable
allocation problem is in assigning personnel with ``two-body''
constraints.  For example, in the NRMP, a married pair of medical
school graduates might act as an unsplittable entity of size 2 (this
particular application has been studied in substantial detail in the
literature
\cite{Biro_emp,couples_survey,journals/jal/Ronn90,Roth84theevolution}).

From an algorithmic standpoint, one of the main results of this paper
is that how we define stability in the unsplit case seems quite
important.  In \cite{DBLP:journals/jco/McDermidM10}, the following
natural definition was proposed: an edge $jm$ is blocking if $j$
prefers $m$ to its current partner, and if $m$ prefers $j$ over $q(j)$
units of its current allocation.  Unfortunately, it was shown in
\cite{DBLP:journals/jco/McDermidM10} that this definition makes the
computation of an unsplit stable assignment NP-hard.  We therefore
consider an alternate, stricter notion of stability where edge $jm$ is
blocking if $j$ prefers $m$ to its current partner, and if $m$ prefers
$j$ over {\em any amount} of its current allocation.  That is, if $j$
would prefer to be assigned to $m$ over its current partner, than $m$
must be saturated with jobs $m$ prefers to $j$.  As in the splittable
case, popular machines must therefore be saturated. Practice shows~\cite{Roth96} that if a hospital is willing to hire one person in a couple, but it has no free job opening for the partner, it is most likely amenable to make room for both applicants. Therefore, our definition of a blocking pair serves practical purposes.

The existence of an unsplit stable allocation cannot be guaranteed. A
simple instance where the unique stable allocation is fractional is
shown in Figure \ref{label5}.  The quota of
each job and machine is displayed next to the vertex, while the preference lists are displayed on the edges.

\begin{figure}[t]
\begin{center}
\begin{tikzpicture}[scale=0.75, transform shape]
\tikzstyle{vertex} = [circle, draw=black]
\tikzstyle{edgelabel} = [rectangle, fill=white]

\node[vertex, label={[label distance=0.1cm]180:{1}}] (j_1) at (0, 0) {$j_1$};
\node[vertex, label={[label distance=0.1cm]0:{1}}] (m_1) at (4, 0) {$m_1$};
\node[vertex, label={[label distance=0.1cm]180:{2}}] (j_2) at (0, 2) {$j_2$};
\node[vertex, label={[label distance=0.1cm]0:{2}}] (m_2) at (4, 2) {$m_2$};

\draw [] (j_2) -- node[edgelabel, near start] {2} node[edgelabel, near end] {1} (m_2);
\draw [] (j_2) -- node[edgelabel, near start] {1} node[edgelabel, near end] {1} (m_1);
\draw [] (j_1) -- node[edgelabel, near start] {1} node[edgelabel, near end] {2} (m_1);

\end{tikzpicture}
\mycaption{label5}{This instance admits an unsplit assignment, but 
the unique stable allocation is fractional.}
\end{center}
\end{figure}

\noindent {\bf Relaxed Unsplit Assignments.}  The downside of our
alternate definition of stability is that it is unlikely to allow
feasible unsplit stable allocations to exist in most large instances.
Therefore, we consider allowing mildly-infeasible solutions where each
machine can be over-capacitated by a single job -- a model popularized
by previous results in the approximation algorithm literature for
standard unsplittable assignment problems
\cite{DinitzGG99,Skutella00,ShmoysT93}, and introduced in the context
of unsplittable stable allocation by Dean et al.~\cite{dean_unsplit}.
Specifically, we say $x$ is a \emph{relaxed unsplit assignment} if
$x(jm) \in \{0,q(j)\}$ for every edge $jm \in E$, and if for each
machine $m$, removal of the least-preferred job assigned to $m$ would
cause $x(m) < q(m)$\footnote{The model introduced in
  \cite{dean_unsplit} allows $x(m) \leq q(m)$, but we believe strict
  inequality is actually a better choice -- for mathematical reasons
  as well as from a modeling standpoint.  For example, the old
  definition applied to a hospital-resident matching scenario with
  married couples might cause a hospital to accept two more residents
  than its quota, while the new definition would only require
  accepting one more resident.  All of the
  results in \cite{dean_unsplit} hold with either definition.}.  Our
definition of stability extends easily naturally to the relaxed
setting: we say a relaxed unsplit assignment $x$ is stable if for
every edge $jm$ with $x(jm)=0$, either $j$ is assigned to a machine
$j$ prefers to $m$, or $m$'s quota is filled or exceeded with jobs $m$
prefers to $j$.  Otherwise, if edge $jm$ with $x(jm) = 0$ is preferred
by $j$ to its allocated machine and $m$'s quota is not filled up with
better edges than~$jm$, then $jm$ \emph{blocks}~$x$.

Note that the relaxed unsplit model differs from the non-relaxed
unsplit model with capacities inflated by $\max q(j)$, since stability
is still defined with respect to the original capacities.  It may be
best to regard ``capacities'' in this setting as constraints governing
start time, rather than completion time of jobs, since a machine below
its capacity is always willing to launch a new job, irrespective of
job size.  Similarly, a machine $m$ views an edge $jm$ as blocking if
the machine is not fully saturated with jobs $m$ prefers to $jm$, as
in this case $m$ is willing to accept $jm$.  The ``capacity'' of a
machine therefore reflects the cutoff at which it feels content to
receive additional assignment versus when it can no longer accept
additional load.

\section{Machine-Optimal Relaxed Unsplit Assignments}
\label{sec:jmopt}

In \cite{dean_unsplit}, a version of the Gale-Shapley algorithm is
described to find the job-optimal relaxed unsplit stable
assignment~$x_{\text{jopt}}$.  In this context, job-optimal means that
there is no relaxed unsplit stable assignment $x'$ such that any job
is assigned to a better machine in $x'$ than in~$x_{\text{jopt}}$.
The implementation described in \cite{dean_unsplit} runs in $O(|E|
\log |V|)$ time, but $O(|E|)$ is also easy to achieve.  In this
section, we show how to define and compute a \emph{machine-optimal}
relaxed unsplit stable allocation $x_{\text{mopt}}$ also in $O(|E|)$
time, and we prove the following:

\begin{theorem}
\label{thm:card}
Among all relaxed unsplit stable allocations $x$, $|x|$ is maximized
at $x = x_{\text{jopt}}$ and minimized at $x = x_{\text{mopt}}$.
\end{theorem}

\begin{figure}[t]
\begin{center}
\begin{minipage}{.3\textwidth}
\begin{tikzpicture}[scale=0.85, transform shape]
	\tikzstyle{vertex} = [circle, draw=black]
	\tikzstyle{edgelabel} = [rectangle, fill=white]

	\node[vertex, label={[label distance=0.1cm]180: {2}}] (j_1) at (0, 0) {$j_1$};
	\node[vertex, label={[label distance=0.1cm]180: {1}}] (j_2) at (0, 2) {$j_2$};
	\node[vertex, label={[label distance=0.1cm]180: {2}}] (j_3) at (0, 4) {$j_3$};
	\node[vertex, label={[label distance=0.1cm]0: {2}}] (m_1) at (3, 1) {$m_1$};
	\node[vertex, label={[label distance=0.1cm]0: {1}}] (m_2) at (3, 3) {$m_2$};

	\draw [dashed] (j_3) -- node[edgelabel, near start] {2} node[edgelabel, near end] {1} (m_1);
	\draw [dashed] (j_2) -- node[edgelabel, near start] {2} node[edgelabel, near end] {1} (m_2);
	\draw [] (j_3) --node[edgelabel, near start] {1} node[edgelabel, near end] {2} (m_2);
	\draw [] (j_2) --  node[edgelabel, near start] {1} node[edgelabel, near end] {2}(m_1);
	\draw [] (j_1) -- node[edgelabel, near start] {1} node[edgelabel, near end] {3} (m_1);
	
\end{tikzpicture}

\vspace{0.8cm}

\begin{tikzpicture}[scale=0.85, transform shape]
	\tikzstyle{vertex} = [circle, draw=black]
	\tikzstyle{edgelabel} = [rectangle, fill=white]

	\node[vertex, label={[label distance=0.1cm]180: {2}}] (j_1) at (0, 0) {$j_1$};
	\node[vertex, label={[label distance=0.1cm]180: {3}}] (j_2) at (0, 2) {$j_2$};
	\node[vertex, label={[label distance=0.1cm]180: {1}}] (j_3) at (0, 4) {$j_3$};
	
	\node[vertex, label={[label distance=0.1cm]0: {1}}] (m_1) at (3, 0) {$m_1$};
	\node[vertex, label={[label distance=0.1cm]0: {3}}] (m_2) at (3, 2) {$m_2$};
	\node[vertex, label={[label distance=0.1cm]0: {1}}] (m_3) at (3, 4) {$m_3$};

	\draw [dashed] (j_3) -- node[edgelabel, near start] {1} node[edgelabel, near end] {2} (m_2);
	\draw [dashed] (j_2) -- node[edgelabel, near start] {1} node[edgelabel, near end] {2} (m_3);
	\draw [] (j_3) -- node[edgelabel, near start] {2} node[edgelabel, near end] {1} (m_3);
	\draw [] (j_2) -- node[edgelabel, near start] {2} node[edgelabel, near end] {1} (m_2);
	\draw [dashed] (j_1) -- node[edgelabel, near start] {1} node[edgelabel, near end] {3} (m_2);
	\draw [] (j_1) -- node[edgelabel, near start] {2} node[edgelabel, near end] {1} (m_1);
\end{tikzpicture}

\end{minipage}\hspace{3mm}\begin{minipage}{.3\textwidth}

\begin{tikzpicture}[scale=0.85, transform shape]
	\tikzstyle{vertex} = [circle, draw=black]
	\tikzstyle{edgelabel} = [rectangle, fill=white]
	\pgfmathsetmacro{\d}{1.8}

	\node[vertex, label={[label distance=0.1cm]180: {1}}] (j_1) at (0, 0) {$j_1$};
	\node[vertex, label={[label distance=0.1cm]180: {2}}] (j_2) at (0, \d) {$j_2$};
	\node[vertex, label={[label distance=0.1cm]180: {2}}] (j_3) at (0, \d*2) {$j_3$};
	\node[vertex, label={[label distance=0.1cm]180: {1}}] (j_4) at (0, \d*3) {$j_4$};
	\node[vertex, label={[label distance=0.1cm]180: {1}}] (j_5) at (0, \d*4) {$j_5$};
	\node[vertex, label={[label distance=0.1cm]180: {2}}] (j_6) at (0, \d*5) {$j_6$};
	\node[vertex, label={[label distance=0.1cm]180: {1}}] (j_7) at (0, \d*6) {$j_7$};
	
	\node[vertex, label={[label distance=0.1cm]0: {1}}] (m_1) at (3, \d*1) {$m_1$};
	\node[vertex, label={[label distance=0.1cm]0: {2}}] (m_2) at (3, \d*2) {$m_2$};
	\node[vertex, label={[label distance=0.1cm]0: {2}}] (m_3) at (3, \d*3) {$m_3$};
	\node[vertex, label={[label distance=0.1cm]0: {2}}] (m_4) at (3, \d*4) {$m_4$};
	\node[vertex, label={[label distance=0.1cm]0: {1}}] (m_5) at (3, \d*5) {$m_5$};
	
	\draw [dashed] (j_1) -- node[edgelabel, near start] {2} node[edgelabel, near end] {1} (m_1);
	\draw [dashed] (j_2) -- node[edgelabel, near start] {2} node[edgelabel, near end] {1} (m_2);
	\draw [dashed] (j_3) -- node[edgelabel, near start] {1} node[edgelabel, near end] {3} (m_3);
	\draw [dashed] (j_4) -- node[edgelabel, near start] {2} node[edgelabel, near end] {2} (m_3);
	\draw [dashed] (j_5) -- node[edgelabel, near start] {1} node[edgelabel, near end] {3} (m_4);
	\draw [dashed] (j_6) -- node[edgelabel, near start] {1} node[edgelabel, near end] {2} (m_5);
	\draw [dashed] (j_7) -- node[edgelabel, near start] {1} node[edgelabel, near end] {2} (m_4);
	
	\draw [] (j_1) -- node[edgelabel, near start] {1} node[edgelabel, near end] {2} (m_2);
	\draw [] (j_2) -- node[edgelabel, near start] {1} node[edgelabel, near end] {2} (m_1);
	\draw [] (j_3) to[out=0,in=-110] node[edgelabel, near start] {1} node[edgelabel, near end] {3}  (m_3);
	\draw [] (j_4) -- node[edgelabel, near start] {1} node[edgelabel, near end] {3} (m_2);
	\draw [] (j_5) -- node[edgelabel, near start] {2} node[edgelabel, near end] {1} (m_3);
	\draw [] (j_6) -- node[edgelabel, near start] {2} node[edgelabel, near end] {1} (m_4);
	\draw [] (j_7) -- node[edgelabel, near start] {2} node[edgelabel, near end] {1} (m_5);
\end{tikzpicture}
\end{minipage}\hspace{3mm}\begin{minipage}{.35\textwidth}
\begin{tikzpicture}[scale=0.85, transform shape]
	\tikzstyle{vertex} = [circle, draw=black]
	\tikzstyle{edgelabel} = [rectangle, fill=white]
	\pgfmathsetmacro{\d}{1.8}
	\pgfmathsetmacro{\b}{3}

	\node[vertex, label={[label distance=0.1cm]180: {1}}] (j_1) at (0, 0) {$j_1$};
	\node[vertex, label={[label distance=0.1cm]180: {2}}] (j_2) at (0, \d) {$j_2$};
	\node[vertex, label={[label distance=0.1cm]180: {1}}] (j_3) at (0, \d*2) {$j_3$};
	\node[vertex, label={[label distance=0.1cm]180: {2}}] (j_4) at (0, \d*3) {$j_4$};
	\node[vertex, label={[label distance=0.1cm]180: {3}}] (j_5) at (0, \d*4) {$j_5$};
	\node[vertex, label={[label distance=0.1cm]180: {1}}] (j_6) at (0, \d*5) {$j_6$};
	\node[vertex, label={[label distance=0.1cm]180: {2}}] (j_7) at (0, \d*6) {$j_7$};
	
	\node[vertex, label={[label distance=0.1cm]0: {1}}] (m_1) at (\b, \d*0.5) {$m_1$};
	\node[vertex, label={[label distance=0.1cm]0: {2}}] (m_2) at (\b, \d*1.5) {$m_2$};
	\node[vertex, label={[label distance=0.1cm]0: {2}}] (m_3) at (\b, \d*2.5) {$m_3$};
	\node[vertex, label={[label distance=0.1cm]0: {3}}] (m_4) at (\b, \d*3.5) {$m_4$};
	\node[vertex, label={[label distance=0.1cm]0: {1}}] (m_5) at (\b, \d*4.5) {$m_5$};
	\node[vertex, label={[label distance=0.1cm]0: {1}}] (m_6) at (\b, \d*5.5) {$m_6$};
	
	\draw [dashed] (j_1) -- node[edgelabel, near start] {2} node[edgelabel, near end] {1} (m_1);
	\draw [dashed] (j_2) -- node[edgelabel, near start] {2} node[edgelabel, near end] {1} (m_2);
	\draw [dashed] (j_3) -- node[edgelabel, near start] {2} node[edgelabel, near end] {2} (m_3);
	\draw [dashed] (j_4) -- node[edgelabel, near start] {1} node[edgelabel, near end] {3} (m_4);
	\draw [dashed] (j_5) -- node[edgelabel, near start] {1} node[edgelabel, near end] {2} (m_5);
	\draw [dashed] (j_6) -- node[edgelabel, near start] {1} node[edgelabel, near end] {2} (m_4);
	\draw [dashed] (j_7) -- node[edgelabel, near start] {2} node[edgelabel, near end] {1} (m_6);
	\draw [dotted] (j_7) -- node[edgelabel, very near start] {1} node[edgelabel, very near end] {3} (m_3);
	
	\draw [] (j_1) -- node[edgelabel, near start] {1} node[edgelabel, near end] {2} (m_2);
	\draw [] (j_2) -- node[edgelabel, near start] {1} node[edgelabel, near end] {2} (m_1);
	\draw [] (j_3) -- node[edgelabel, near start] {1} node[edgelabel, near end] {3} (m_2);
	\draw [] (j_4) -- node[edgelabel, near start] {2} node[edgelabel, near end] {1} (m_3);
	\draw [] (j_5) -- node[edgelabel, near start] {2} node[edgelabel, near end] {1} (m_4);
	\draw [] (j_6) -- node[edgelabel, near start] {2} node[edgelabel, near end] {1} (m_5);
	\draw [] (j_7) to[in=130, out=30] node[edgelabel, near start] {2} node[edgelabel, near end] {1}  (m_6);
	
\end{tikzpicture}
\end{minipage}
\mycaption{label6}{The upper-left instance admits two relaxed unsplit solutions differing in
cardinality.  The lower-left example is evidence against a rural hospital theorem.  The graph in the middle shows two incomparable relaxed unsplit solutions. The last instance is a counterexample showing the difficulty of formulating join and meet operations. The first and third graphs illustrate instances of NRMP.}
\end{center}
\end{figure}

One of the main challenges with computing a machine-optimal assignment
is defining machine-optimality.  In the stable allocation problem,
existence of a machine-optimal solution follows from the fact that all
stable solutions form a distributive lattice under the standard
min-min ordering relationship introduced in
\cite{DBLP:journals/mor/BaiouB02}.  However, this ordering seems to
depend crucially on the existence of a rural hospital theorem, which
no longer holds in the relaxed unsplit case, since relaxed unsplit
stable assignments may differ in cardinality, as shown in the
upper-left example in Figure \ref{label6}.  The dashed edges form a
stable solution of size~3, while the remaining edges build another
stable solution of size~6. Even an appropriately relaxed version of
the rural hospital theorem seems difficult to formulate over 
relaxed instances: machines can be saturated or even over-capacitated
in one relaxed unsplit stable solution, while being empty in another
one. The lower-left example in the figure shows such an instance: the
two stable assignments are denoted with the same line types, and $m_1$
is the machine that has different positions in them.  Nonetheless,
we can still prove a result in the spirit of the rural hospital theorem,
which we discuss further in Section \ref{rural_hospital}.

Without an ``exact'' rural hospital theorem, comparing two allocations
using the original min-min ordering seems problematic, and indeed one
can construct instances where two relaxed unsplit stable solutions are
incomparable according to this criterion.  For example, the instance
on the right in Figure \ref{label6} shows two relaxed unsplit
solutions (indicated with dotted and solid edges) that are
incomparable for machine~$m_3$. We therefore adopt a different but
nonetheless natural ordering relation: \emph{lexicographical
  order}. We say that machine $m$ prefers unsplit allocation $x_1$ to
allocation $x_2$ if the best edge in $x_1 \triangle x_2$ belongs
to~$x_1$, where $\triangle$ denotes the symmetric difference
operation.  The opposite ordering relation is based on the position of
jobs, and since jobs are always assigned to machines in an unsplit
fashion, the lexicographic and min-min relations are actually the same
from the job's perspectives; hence, ``job optimal'' means the same
thing under both.  The lexicographical position of the same agent in
different allocations can always be compared, and we say a relaxed
stable solution $x$ is \emph{machine-optimal} if it is at least as
good for all machines as any other relaxed stable assignment (although
we still need to show that such a solution always exists).

\subsection{The Reversed Gale-Shapley Algorithm}

For the classical stable marriage problem, the Gale-Shapley algorithm
can be reversed easily, with women proposing instead of men, to obtain
a woman-optimal solution.  We show that this idea can be generalized
(carefully accounting for multiple assignment and congestion among
machines) to compute a machine-optimal relaxed unsplit stable
assignment.  Pseudocode for the algorithm appears in Figure \ref{rev_gs}.

\begin{figure}[t]
\begin{center}
\begin{algorithmic}[1]
	\State $x(jm_d) := q(j)$ for all $j \in J$, $x(jm) := 0$ for every other $jm \in E$ 
	\While{$\exists m: x(m) < q(m)$ with a non-empty preference list}
		\State $m$ proposes to its best job $j$ with $q(j)$
		\If{$j$ prefers $m$ to its current partner}
			\State $x(jm) := q(j)$
			\State $x(jm') := 0$ for $\forall m' \neq m$
		\EndIf
			\State delete $j$ from $m$'s preference list
	\EndWhile
\end{algorithmic}
\vspace*{-0.1in}
\end{center}
\mycaption{rev_gs}{Reversed relaxed unsplit Gale-Shapley algorithm.}
\end{figure}

\begin{claim}
The algorithm terminates in $O(|E|)$ time.
\end{claim}

\begin{proof}
In each step, a job is deleted from a machine's preference
list. 
\end{proof}

\begin{claim}
The algorithm produces an allocation $x$ that is a relaxed unsplit
stable assignment.
\end{claim}

\begin{proof}
First, we check the three feasibility constraints for~$x$. Since
proposals are always made with $q(j)$ and refusals are always full
rejections, the quota constraints of the jobs may not be
violated. Moreover, each job is assigned to exactly one
machine. Machines can be over-capacitated, but deleting the worst job
from their preference list results in an allocation under their quota. Otherwise the machine would not have proposed along
the last edge.
	
If $x$ is unstable, then there is an empty edge $jm$
blocking~$x$. During the execution, $m$ must have proposed
to~$j$. This offer was rejected, because $j$ already had a better
partner in the current allocation. Since jobs monotonically improve
their position in the assignment, this leads to a contradiction.
\end{proof}

\begin{claim}
The output $x$ is the machine-optimal relaxed unsplit stable
assignment. That is, no machine has a better lexicographical position
in any other relaxed unsplit stable assignment than in~$x$.
\end{claim}

\begin{proof}
Assume that there is a relaxed unsplit stable assignment~$x'$, where
some machines come better off than in~$x$. To be more precise, in the
symmetric difference $x \triangle x'$, the best edge incident to these
machines belongs to~$x'$. When running the reversed relaxed unsplit
Gale-Shapley algorithm, there is a step when the first such edge
$jm_1$ carries a proposal from $m_1$ but gets rejected. Otherwise,
$m_1$ filled up or exceeded its quota in $x$ with only better edges
than~$jm_1$. Let us consider only this edge first and denote the
feasible, but possibly unstable relaxed allocation produced by the
algorithm so far by~$x_0$.

When $j$ refused~$jm_1$, it already had a partner $m_0$ in~$x_0$,
better than~$m_1$. Even if there is no guarantee that $jm_0 \in x$, it
is sure that $jm_0 \notin x'$ and $jm_0$ does not block~$x'$, though
$\text{rank}_j(jm_0) > \text{rank}_j(jm_1)$ for~$jm_1 \in x'$. It is
only possible if $m_0$ is saturated or over-capacitated in $x'$ with
edges better than~$jm_0$. Since $jm_0 \in x_0$, $x_0$ may not contain
all of these edges, otherwise $m_0$ is congested in $x_0$ beyond the
level required for a relaxed unsplit assignment. During the execution
of the reversed relaxed unsplit Gale-Shapley algorithm, $m_0$ proposed
along all of these edges and got rejected by at least one of
them. This edge is never considered again, it may not enter $x$
later. Thus, $jm_1$ is not the first edge in $x' \setminus x$ that was
rejected in the algorithm.
\end{proof}

With this, we completed the constructive proof of the following theorem:

\begin{theorem}
\label{m_opt}
The machine-optimal relaxed unsplit stable assignment
$x_{\text{mopt}}$ can be computed in $O(|E|)$ time.
\end{theorem}

\subsection{Properties of the Job- and Machine-Optimal Solutions}

\begin{theorem}
\label{th:opt_pess}
The job-optimal relaxed unsplit stable assignment $x_{\text{jopt}}$ is
the machine-pessimal relaxed unsplit stable assignment and vice versa,
the machine-optimal relaxed unsplit stable assignment
$x_{\text{mopt}}$ is the job-pessimal relaxed unsplit stable
assignment.
\end{theorem}

\begin{proof}

We start with the first statement. Suppose that there is a relaxed
unsplit stable assignment $x'$ that is worse for some machine $m$
than~$x_{\text{jopt}}$.  This is only possible if $m$'s best edge $jm$
in $x_{\text{jopt}} \triangle x'$ belongs to~$x_{\text{jopt}}$. Since
$x_{\text{jopt}}$ is the job-optimal solution, $jm'$, $j$'s edge in
$x'$ is worse than~$jm$. But then, $m$ is saturated or
over-capacitated in $x'$ with better edges than~$jm$. We assumed that
all edges in $x'$ that are better than $jm$ are also
in~$x_{\text{jopt}}$. Thus, omitting $m$'s worst job from
$x_{\text{jopt}}$ leaves $m$ at or over its quota.

The second half of the theorem can be proved similarly, using the
reversed Gale-Shapley algorithm. Assume that there is a relaxed
unsplit stable assignment $x'$ that assigns some jobs to worse
machines than~$x_{\text{mopt}}$ does. Let us denote the set of edges
preferred by any job to its allocated machine in $x'$ by~$E'$. Due to
our indirect assumption, $E'$ contains some edges
of~$x_{\text{mopt}}$. When running the reversed Gale-Shapley algorithm
on the instance, there is an edge $jm \in E'$ that is the first edge
in $E'$ carrying a proposal. Since $j$ is not yet matched to a better
machine, it also accepts this offer. Even if $jm \notin
x_{\text{mopt}}$, $j$'s edge in $x_{\text{mopt}}$ is at least as good
as~$m$, because jobs always improve their position during the course
of the reversed Gale-Shapley algorithm. On the other hand, $m$ cannot
fulfill its quota in $x_{\text{mopt}}$ with better edges than~$jm$,
simply because the proposal step along $jm$ took place.

Since $jm \notin x'$, but $j$ prefers $jm$ to its edge in $x'$, $m$ is
saturated or over-capacitated with better edges than $jm$ in~$x'$. As
observed above, not all of these edges belong
to~$x_{\text{mopt}}$. Let us denote one of them in $x' \setminus
x_{\text{mopt}}$ by~$j'm$. Before proposing along $jm$, $m$ submitted
an offer to $j'$ that has been refused. The only reason for such a
refusal is that $j'$ has already been matched to a better
machine~$m'$. But since $j'm \in x'$, $j'm' \in E'$. This contradicts
to our indirect assumption that $jm$ is the first edge in $E'$ that
carries a proposal. \end{proof}

Theorem \ref{thm:card} also follows from the proof above.  

We note that although we can compute the job-optimal and
machine-optimal relaxed unsplit stable allocations, there in general
does not appear to be an obvious underlying lattice structure behind
relaxed unsplit solutions.  As already mentioned above, stable
matchings and fractional stable allocations form a distributive
lattice. In those cases, computing the meet or join of two solutions
is fairly easy. In order to reach the join of $x_1$ and $x_2$, all
machines choose the better edge set out of those two
allocations~\cite{egres-09-11}. Similarly, for meet, 
jobs are allowed to chose.  The example in Figure \ref{label3}
illustrates that this property does not carry over to relaxed unsplit
assignments. If all jobs chose the better allocation, $m_3$ remains
empty and $j_7m_3$ becomes blocking. Similar examples can easily be
constructed to show that choosing the worse assignment also can lead
to instability.

Our ability to compute $x_{\text{mopt}}$ in $O(|E|)$ time now gives us
a linear-time method for solving the (non-relaxed) unsplittable stable
allocation problem (according to our stricter notion of stability).

\begin{lemma}
\label{le:mopt_unspl}
If an instance $\mathcal{I}$ admits an unsplit stable assignment~$x$,
then the machine-optimal relaxed unsplit stable assignment
$x_{\text{mopt}}$ on the corresponding relaxed instance $\mathcal{I}'$
is also an unsplit stable assignment on~$\mathcal{I}$.
\end{lemma}

\begin{proof}
Suppose the statement is false, e.g.\ although there is an unsplit
stable assignment~$x$, $x_{\text{mopt}}$ is no unsplit stable
assignment on~$\mathcal{I}$. This can be due to two reasons: either
the feasibility or the stability of $x_{\text{mopt}}$ is harmed
on~$\mathcal{I}$. The latter case is easier to handle. An allocation
that is feasible on both instances and stable on $\mathcal{I}'$ may
not be blocked by any edge on~$\mathcal{I}$, since the set of
unsaturated edges is identical on both instances. The second case,
namely if $x_{\text{mopt}}$ violates some feasibility constraint
on~$\mathcal{I}$, needs more care.
	
$\mathcal{I}$ and $\mathcal{I'}$ differ only in the constraints on the
quota of machines. If $x_{\text{mopt}}$ is infeasible
on~$\mathcal{I}$, then there is a machine $m$ for which
$x_{\text{mopt}}(m_1) > q(m_1)$. Regarding the unsplit stable
assignment~$x$, the inequality $x(m_1) \leq q(m_1)$ trivially
holds. Now we use Theorem~\ref{thm:card} for $x$ and $x_{\text{mopt}}$
that are both relaxed unsplit stable assignments
on~$\mathcal{I}'$. This corollary implies that if there is a machine
$m_1$ with $x_{\text{mopt}}(m_1) > x(m_1)$, then another machine $m_2$
exists for which $x_{\text{mopt}}(m_2) < x(m_2)$ holds.
	
This machine $m_2$ plays a crucial role in our proof. It has a lower
allocation value in the machine-optimal relaxed solution $x_{\text{mopt}}$
than in another relaxed stable solution $x$ on~$\mathcal{I}$. Its
lexicographical position can only be better in $x_{\text{mopt}}$ than
in $x$ if the best edge $j_2 m_2$ in $x \triangle x_{\text{mopt}}$
belongs to~$x_{\text{mopt}}$. Moreover, $x \triangle x_{\text{mopt}}$
also contains an edge $j_3 m_2 \in x$, otherwise $x_{\text{mopt}}(m_2)
> x(m_2)$. Naturally, $\text{rank}_m(j_2 m_2) < \text{rank}_m(j_3
m_2)$. At this point, we use the property that $x_{\text{mopt}}(m_2) <
q(m_2)$. Since $m_2$ has free quota in $x_{\text{mopt}}$ and $j_3 m_2$
is not a blocking edge, $j_3$ must be matched to a machine better than
$m_2$ in~$x_{\text{mopt}}$. Thus, there is a job that comes better off
in the machine-optimal (and job-pessimal) relaxed solution than in another
relaxed stable solution. This contradiction to
Theorem~\ref{th:opt_pess} finishes our proof.
\end{proof}


Lemma~\ref{le:mopt_unspl} shows that if there is an unsplit solution,
it can be found in linear time by computing the machine-optimal
relaxed solution. Unfortunately, the existence of such an unsplit
assignment is not guaranteed. Our next result applies to the case when
no feasible unsplit solution can be found. In terms of congestion,
with the machine-optimal solution we come as close as possible to
feasibility.

\begin{theorem}
	Amongst all relaxed unsplit stable solutions, $x_{\text{mopt}}$ has the lowest total congestion.
\end{theorem}

\begin{proof}
	Let $M_u$ denote the set of machines that remain under their quota in~$x_{\text{mopt}}$. Note that $\sum_{m \notin M_u} {x_{\text{mopt}}(m)}$, the total allocation value on the remaining machines clearly determines $\sum_{m \notin M_u} {x_{\text{mopt}}(m) - q(m)}$, the total congestion of~$x_{\text{mopt}}$. Let $x$ be an arbitrary relaxed solution. Due to Theorem~\ref{thm:card}, the total allocation value is minimized at~$x_{\text{mopt}}$. Therefore, for any relaxed unsplit stable allocation $x$, the following inequalities hold:
	\begin{align*}
	\sum_{m \in M}{x(m)} &\geq \sum_{m \in M}{x_{\text{mopt}}(m)}\\
	\sum_{m \notin M_u} {x(m)} + \sum_{m \in M_u} {x(m)} &\geq \sum_{m \notin M_u} {x_{\text{mopt}}(m)} + \sum_{m \in M_u} {x_{\text{mopt}}(m)}\\
	\sum_{m \notin M_u} {x(m)} - \sum_{m \notin M_u} {x_{\text{mopt}}(m)}&\geq \sum_{m \in M_u} {x_{\text{mopt}}(m)} - \sum_{m \in M_u} {x(m)}\\
	\sum_{m \notin M_u} {(x(m)-q(m))} - \sum_{m \notin M_u} {(x_{\text{mopt}}(m)-q(m))}&\geq \sum_{m \in M_u} {x_{\text{mopt}}(m)} - \sum_{m \in M_u} {x(m)}
	\end{align*}
	
At this point, we investigate the sign of both sides of the last
inequality. The core of our proof is to show that for each $m \in M_u$
and relaxed stable solution~$x$, $x_{\text{mopt}}(m) \geq x(m)$. This
result, proved below, has two benefits. On one hand, the term
on the right hand-side of the last inequality is
non-negative. Therefore, the inequality implies that the total
congestion on machines in $M \setminus M_u$ is minimized at~$x_{\text{mopt}}$. On the other hand, no machine in $M_u$ is
over-capacitated in any relaxed solution. Thus, the total congestion
is minimized~at~$x_{\text{mopt}}$.
\end{proof}

\begin{lemma}
\label{rh_part1}
 For every $m \in M_u$ and relaxed solution~$x$, the inequality
 $x_{\text{mopt}}(m) \geq x(m)$ holds.
\end{lemma}
	
\begin{proof}
Suppose that there is a machine $m \in M_u$ for which
$x_{\text{mopt}}(m) < x(m)$ for some relaxed solution~$x$. Since $m$
is unsaturated in $x_{\text{mopt}}$, it is unpopular. On the other
hand, there is at least one job $j$ for which $jm \in x \setminus
x_{\text{mopt}}$. As $m$ is unpopular in~$x_{\text{mopt}}$, $j$ is
allocated to a better machine in $x_{\text{mopt}}$ than in~$x$. Since
$x_{\text{mopt}}$ is the job-pessimal solution, we derived a
contradiction.
\end{proof}

\subsection{A Variant of the ``Rural Hospital'' Theorem}
\label{rural_hospital}
In the relaxed unsplit case, we have provided counterexamples against
an exact rural hospital theorem (e.g., where all machines have the
same amount of allocation in all relaxed unsplit allocations) or even
a weakened theorem stating that all unsaturated / congested machines
have the same status in all relaxed unsplit allocations. The examples in Figure~\ref{label6} also show that no similar property holds for the jobs' side either. Lemma~\ref{rh_part1} above however suggests an alternate variant of ``rural
hospital'' theorem that does hold.

\begin{theorem}
A machine $m$ that is not saturated in $x_{\text{mopt}}$ will not be
saturated in every relaxed unsplit stable solution, and a machine $m$
that is over-capacitated in $x_{\text{jopt}}$ must at least be
saturated in every relaxed unsplit stable solution.
\end{theorem}

\begin{proof}
The first part is shown by Lemma \ref{rh_part1}.  For the second part,
consider a machine $m$ that is over-capacitated in $x_{\text{jopt}}$
but has $x(m) < q(m)$ in some relaxed unsplit allocation $x$.
Consider any job $j$ in $x_{\text{jopt}} \backslash x$, and note that
since $x_{\text{jopt}}$ is job-optimal, $j$ prefers $m$ to its partner
in $x$.  Hence, $jm$ blocks $x$.
\end{proof}

As of the jobs' side, Theorem~\ref{th:opt_pess} already guarantees that if a job is unmatched in $x_{\text{jopt}}$, then it is unmatched in all relaxed stable solutions and similarly, if it is matched in $x_{\text{mopt}}$, then it is matched in all relaxed stable solutions.


\section{Rounding Algorithms}
\label{sec:rot}

We have seen now how to compute $x_{\text{jopt}}$ and
$x_{\text{mopt}}$ in linear time.  We now describe how to find
potentially other relaxed unsplit solutions by ``rounding'' solutions
to the (fractional) stable allocation problem.  For example, this
could provide a heuristic for generating relaxed unsplit solutions
that are more balanced in terms of fairness between the jobs and
machines.  Our approach is based on augmentation around
\emph{rotations}, alternating cycles that are commonly used in stable
matching and allocation problems to move between different stable
solutions (see, e.g., \cite{DBLP:journals/algorithmica/DeanM10,GusfieldI89}).

We begin with a stable allocation $x$ with $x(j) = q(j)$ for every job
$j$, thanks to the existence of a dummy machine.  For each job $j$
that is not fully assigned to its first-choice machine, we define its
\emph{refusal edge} $r(j)$ to be the worst edge $jm$ incident to $j$
with $x(jm) > 0$. Jobs with refusal edges also have \emph{proposal
  edges} -- namely all their edges ranked better than $r(j)$.  Recall
that a machine with incoming proposal edges is said to be
\emph{popular}.  We call a machine {\em dangerous} if it is
over-capacitated and has zero assignment on all its incoming proposal
edges.

\begin{claim}
\label{claim:structure_of_h}
Consider a popular machine $m$ in some fractional stable
allocation~$x$. Amongst all proposal edges incoming to $m$, at most
one has positive allocation value in $x$, and this positive proposal
edge is ranked lower on $m$'s preference list than any other edge
into $m$ with positive allocation.
\end{claim}

\begin{proof} Let $\text{rank}_m(j_1m) > \text{rank}_m(j_2m)$ be proposal 
edges such that $x(j_1m)$ and $x(j_2m)$ are both positive. Note that
$j_1m$ blocks $x$, since $j_1$ and $m$ have worse allocated edges
in~$x$. A similar argument implies the last part of the claim.
\end{proof}

Our algorithm proceeds by a series of augmentations around rotations,
defined as follows.  We start from a popular, non-dangerous machine
$m$ (if no such machine exists, the algorithm terminates, having
reached an unsplit solution).  Since $m$ is popular and non-dangerous,
it has incoming proposal edges with positive allocation, and due to
the preceding claim, it must have exactly one such edge $jm$.  We
include $jm$ as well as $j$'s refusal edge $jm'$ in our partial
rotation, then continue building the rotation from $m'$ (again finding
an incoming proposal edge, etc.).  We continue until we close a cycle,
visiting some machine $m$ visited earlier (in which case we
keep just the cycle as our rotation, not the edges leading up to the
cycle), or until we reach a machine $m$ that is unpopular or
dangerous, where our rotation ends.

To enact a rotation, we increase the allocation on its proposal edges
by $\varepsilon$ and decrease along the refusal edges by
$\varepsilon$, where $\varepsilon$ is chosen to be as large as
possible until either (i) a refusal edge along the rotation reaches
zero allocation, or (ii) a dangerous machine at the end of the
rotation drops down to being exactly saturated from being
over-capacitated, and hence ceases to be dangerous.  We call case (i)
a ``regular'' augmentation.  This concludes the algorithm description.

\begin{claim}
\label{cl:rot_term}
The algorithm terminates after $O(|E|)$ augmentations.
\end{claim}

\begin{proof} 
Jobs remain fully allocated during the whole procedure, and their
lexicographical positions never worsen.  With every regular
augmentation, some edge stops being a refusal edge, and will never
again be increased or serve as a proposal or refusal edge.  We can
therefore have at most $O(|E|)$ regular augmentations.  Furthermore, a
machine can only become dangerous if one of its incoming refusal
pointers reaches zero allocation, so the number of newly-created
dangerous machines over the entire algorithm is bounded by $|E|$.
Hence, the number of non-regular augmentations is at most $O(|M| +
|E|) = O(|E|)$.
\end{proof}

\begin{claim}
The final allocation $x$ is a feasible relaxed unsplit assignment.
\end{claim}
	
\begin{proof} 
Since we start with a feasible assignment and jobs never lose or gain
allocation, the quota condition on jobs cannot be harmed. If there is
any edge $jm$ with $0 < x(jm) < q(j)$, then $j$ has at least two
positive edges, the better one must be a positive proposal edge. This
contradicts the termination condition, and hence $x$ is unsplit.
	
We now show that deleting the worst job from each machine results in
an allocation strictly below the machine's quota. It is clearly true
at the beginning, where no machine is over-capacitated (since $x$
starts out as a feasible stable allocation). The only case when $x(m)$
increases is when $m$ is the first machine on a rotation. As such, $m$
has a positive proposal edge $jm$, which is also its worst allocated
edge, due to our earlier claim.  
\begin{itemize}
\item If $m$ is not over-capacitated when choosing the
  rotation, then even if $x(jm)$ rises as high as $q(j)$, this
  increases $x(m)$ by strictly less than $q(j)$. Thus, deleting~$jm$,
  the worst allocated edge of $m$, guarantees that $x(m)$ sinks
  under~$q(m)$.
\item If $m$ is saturated or over-capacitated when choosing the
  rotation, then $jm$ would have been the best proposal edge of $m$
  earlier, when $x(m)$ was not greater than~$q(m)$. Thus, assigning
  $j$ entirely to $m$ does not harm the relaxed quota condition. Let
  us consider the last step as $x(m)$ exceeded~$q(m)$. Again, $m$ was
  the starting vertex of an augmenting path, having a positive
  proposal edge. If it was~$jm$, our claim is proved. Otherwise $m$
  became over-capacitated while $x(jm)$ was zero, and then increased
  the allocation on~$jm$. But between those two operations, $m$
  had to become dangerous, because it switched its best proposal edge
  to~$jm$. Dangerous machines never start alternating paths. Thus, we
  have a contradiction to the fact that we considered the last step
  when $x(m)$ exceeded~$q(m)$.
\end{itemize}\end{proof}

\begin{claim}
The final allocation $x$ is stable.
\end{claim}

\begin{proof} 
Suppose some edges block~$x$. Since we started with a stable
allocation, there was a step during the execution of the algorithm
when the first edge $jm$ became blocking. Before this step, either $j$
or $m$ was saturated or over-capacitated with better edges
than~$jm$. The change can be due to two reasons:
\begin{enumerate}
	\item $j$ gained allocation on an edge worse than $jm$, or
	\item $m$ gained allocation on an edge worse than~$jm$.
\end{enumerate}
As already mentioned, $j$'s lexicographical position never worsens:
$\text{rank}_j(p) > \text{rank}_j(r(j))$ always holds. The second
event also may not occur, because machines always play their best
response strategy. An edge $jm$ that becomes blocking when allocation
is increased on an edge worse than it, was already a proposal edge
before. Thus, $m$ would have chosen~$jm$, or an edge better than $jm$
to add it to the augmenting path.
\end{proof}

Since each augmentation requires $O(|V|)$ time and there are $O(|E|)$
augmentations, our rounding algorithm runs in $O(|E||V|)$ total time.
If desired, dynamic tree data structures can be used (much like in
\cite{DBLP:journals/algorithmica/DeanM10}) to augment in $O(\log |V|)$
time, bringing the total time down to just $O(|E| \log |V|)$.  

Although jobs improve their lexicographical position in each rotation,
the output of the algorithm is not necessarily~$x_{\text{jopt}}$. In
fact, even $x_{\text{mopt}}$ can be reached via this approach.
Ideally, this approach can serve as a heuristic to generate many other
relaxed unsplit stable allocations, if run from a variety of different
initial stable solutions $x$.

\bibliographystyle{plain}
\bibliography{mybib2}

\begin{thebibliography}{10}

\bibitem{DBLP:journals/mor/BaiouB02}
Mourad Ba\"{\i}ou and Michel Balinski.
\newblock The stable allocation (or ordinal transportation) problem.
\newblock {\em Math. Oper. Res.}, 27(3):485--503, 2002.

\bibitem{Biro_emp}
P{\'e}ter Bir\'{o}, Robert~W. Irving, and Ildik\'{o} Schlotter.
\newblock Stable matching with couples: An empirical study.
\newblock {\em J. Exp. Algorithmics}, 16:1.2:1.1--1.2:1.27, May 2011.

\bibitem{couples_survey}
P{\'e}ter Bir\'{o} and Flip Klijn.
\newblock Matching with couples: A multidisciplinary survey.
\newblock {\em International Game Theory Review (IGTR)}, 15(02):1340008--1--1,
  2013.

\bibitem{Caprara2000111}
Alberto Caprara, Hans Kellerer, and Ulrich Pferschy.
\newblock A \{PTAS\} for the multiple subset sum problem with different
  knapsack capacities.
\newblock {\em Information Processing Letters}, 73(3-4):111 -- 118, 2000.

\bibitem{DBLP:journals/siamcomp/ChekuriK05}
Chandra Chekuri and Sanjeev Khanna.
\newblock A polynomial time approximation scheme for the multiple knapsack
  problem.
\newblock {\em SIAM J. Comput.}, 35(3):713--728, 2005.

\bibitem{dean_unsplit}
Brian~C. Dean, Michel~X. Goemans, and Nicole Immorlica.
\newblock The unsplittable stable marriage problem.
\newblock In Gonzalo Navarro, Leopoldo~E. Bertossi, and Yoshiharu Kohayakawa,
  editors, {\em IFIP TCS}, volume 209 of {\em IFIP}, pages 65--75. Springer,
  2006.

\bibitem{DBLP:journals/algorithmica/DeanM10}
Brian~C. Dean and Siddharth Munshi.
\newblock Faster algorithms for stable allocation problems.
\newblock {\em Algorithmica}, 58(1):59--81, 2010.

\bibitem{DinitzGG99}
Y.~Dinitz, N.~Garg, and M.X. Goemans.
\newblock On the single-source unsplittable flow problem.
\newblock {\em Combinatorica}, 19:17--41, 1999.

\bibitem{egres-09-11}
Tam\'{a}s Fleiner.
\newblock On stable matchings and flows.
\newblock In {\em Proceedings of the 36th International Workshop on
  Graph-Theoretic Concepts in Computer Science}, WG'10, pages 51--62, Berlin,
  Heidelberg, 2010. Springer-Verlag.

\bibitem{GS:1962}
David Gale and Lloyd Shapley.
\newblock College admissions and the stability of marriage.
\newblock {\em American Mathematical Monthly}, 1:9--14, 1962.

\bibitem{GusfieldI89}
D.~Gusfield and R.W. Irving.
\newblock {\em The Stable Marriage Problem: Structure and Algorithms}.
\newblock MIT Press, 1989.

\bibitem{Knuth:1976:MSR}
Donald~E. Knuth.
\newblock {\em Mariages stables et leurs relations avec d'autres probl{\`e}mes
  combinatoires}.
\newblock Collection de la Chaire Aisenstadt. Les Presses de
  l'Uni\-ver\-sit{\'e} de Montr{\'e}al, Montr{\'e}al, Qu{\'e}bec, Canada, 1976.
\newblock Edition revue et corrig{\'e}e, 1981. Currently available from Les
  Publications CRM / Centre de Recherches Math{\'e}matiques, Universit{\'e} de
  Montr{\'e}al, Montr{\'e}al, Qu{\'e}bec.

\bibitem{DBLP:journals/jco/McDermidM10}
Eric McDermid and David Manlove.
\newblock Keeping partners together: algorithmic results for the
  hospitals/residents problem with couples.
\newblock {\em J. Comb. Optim.}, 19(3):279--303, 2010.

\bibitem{journals/jal/Ronn90}
Eytan Ronn.
\newblock {NP}-complete stable matching problems.
\newblock {\em J. Algorithms}, 11(2):285--304, 1990.

\bibitem{Roth84theevolution}
Alvin~E. Roth.
\newblock The evolution of the labor market for medical interns and residents:
  A case study in game theory.
\newblock {\em Journal of Political Economy}, 92:991--1016, 1984.

\bibitem{Roth96}
Alvin~E. Roth.
\newblock The national residency matching program as a labor market.
\newblock {\em Journal of the American Medical Association},
  275(13):1054--1056, 1996.

\bibitem{ShmoysT93}
D.B. Shmoys and \'{E}. Tardos.
\newblock Scheduling unrelated machines with costs.
\newblock In {\em Proceedings of the 4th annual ACM-SIAM Symposium on Discrete
  algorithms (SODA)}, pages 448--454, 1993.

\bibitem{Skutella00}
M.~Skutella.
\newblock Approximating the single source unsplittable min-cost flow problem.
\newblock In {\em Proceedings of the 41st Annual Symposium on Foundations of
  Computer Science (FOCS)}, pages 136--145, 2000.

\end{thebibliography}
\end{document}